%% file: teiparti.tex
\newcommand{\typeof}{0}

\documentclass[10pt,a4paper]{article}

\usepackage{ifthen}
\newcommand{\condinc}[2]{\ifthenelse{\equal{\typeof}{0}}{#1}{#2}}

\input{environment}

\begin{document}
\input{title}

\tableofcontents
\input{introduction}
\input{proof-nets}
\input{soundness}
\input{completeness}
\input{extensions}

\input{conclusions}

\bibliographystyle{plain}
\bibliography{teiparti}

\end{document}

%% file: environment.tex
\condinc{\usepackage{a4wide}}{}
\usepackage{latexsym}
\usepackage{amsmath,amsfonts,amssymb}
\usepackage{subfigure}
\usepackage{stmaryrd}

\usepackage{epsfig}

\newcommand{\PR}{\textsf{PR}}
\newcommand{\diagonal}{\nabla}
\newcommand{\erase}{\varepsilon}
\newcommand{\proj}{\pi}
\newcommand{\Succ}{s}

\newcommand{\recutel}{\leadsto}
\newcommand{\recutelExp}{\leadsto_{\mathit{XWND}}}

\newcommand{\dsize}[2]{||#1||_{#2}}
\newcommand{\dlength}[2]{[#1]_{#2}}

\newcommand{\B}{\mathcal B}

\newenvironment{varitemize}
{
  \begin{list}{\condinc{\labelitemi}{\labelitemii}}
{
\setlength{\itemsep}{0pt}
 \setlength{\topsep}{0pt}
 \setlength{\parsep}{0pt}
 \setlength{\partopsep}{0pt}
 \setlength{\leftmargin}{15pt}
 \setlength{\rightmargin}{0pt}
 \setlength{\itemindent}{0pt}
 \setlength{\labelsep}{5pt}
 \setlength{\labelwidth}{10pt}}}
{
 \end{list}
}

\newenvironment{varnumlist}
{
\begin{list}{\arabic{number}.}
{\usecounter{number}
 \setlength{\itemsep}{0.0mm}
 \setlength{\topsep}{0.0mm}
 \setlength{\parindent}{0.0mm}
 \setlength{\parskip}{0.0mm}
 \setlength{\parsep}{0.0mm}
 \addtolength{\leftmargin}{-\labelsep}}}
{
 \end{list}
}

\newcommand{\eel}{\end{list} \vspace{\parskip}}

\newcommand{\linear}{\multimap}

\newcommand{\tuple}[1]{\langle #1\rangle}

\newcommand{\prcom}[2]{\circ[#1,#2]}
\newcommand{\prrec}[2]{\mathsf{rec}[#1,#2]}
\newcommand{\funone}{f}
\newcommand{\funtwo}{g}
\newcommand{\funthree}{h}
\newcommand{\prtopn}[1]{G_{#1}}

\newcommand{\N}{\mathbb{N}}

\newcommand{\SLL}{\textsf{SLL}}

\condinc{
\newenvironment{proof}[1][Proof.]{\begin{trivlist}
       \item[\hskip \labelsep {\bfseries #1}]}{\end{trivlist}}
\newenvironment{lemma*}[1][Lemma]{\begin{trivlist}
       \item[\hskip \labelsep {\bfseries Lemma #1}]\it}{\end{trivlist}}
\newenvironment{proposition*}[1][Proposition]{\begin{trivlist}
       \item[\hskip \labelsep {\bfseries Proposition #1}]\it}{\end{trivlist}}
\newenvironment{theorem*}[1][Theorem]{\begin{trivlist}
       \item[\hskip \labelsep {\bfseries Theorem #1}]\it}{\end{trivlist}}

\newtheorem{theorem}{Theorem}[section]
\newtheorem{lemma}[theorem]{Lemma}
\newtheorem{proposition}[theorem]{Proposition}
\newtheorem{corollary}[theorem]{Corollary}
\newtheorem{remark}[theorem]{Remark}

}{}

\newcounter{number}

%% file: title.tex
\title{Taming Modal Impredicativity:\\ Superlazy Reduction}
\condinc{
\author{Ugo Dal Lago\footnote{
Dip. di Informatica, Univ. di Bologna,
\texttt{http://www.cs.unibo.it/\~{}dallago/}
}, Luca Roversi\footnote{
Dip. di Informatica, Univ. di Torino,
\texttt{http://www.di.unito.it/\~{}rover/}
} and Luca Vercelli\footnote{
Dip. di Matematica, Univ. di Torino,
\texttt{http://www.di.unito.it/\~{}vercelli/}
}}
}
{
\author{Ugo Dal Lago\inst{1} \and Luca Roversi\inst{2} \and Luca Vercelli\inst{3}}
\authorrunning{Dal Lago, Roversi, Vercelli}
\institute{Dipartimento di Informatica, Universit\`a di Bologna
\email{http://www.cs.unibo.it/\homedir dallago/}
\and Dipartimento di Informatica, Universit\`a di Torino
\email{http://www.di.unito.it/\homedir rover/}
\and Dipartimento di Matematica, Universit\`a di Torino
\email{http://www.di.unito.it/\homedir vercelli/}
}}

\date{}
\maketitle
\begin{abstract}
Pure, or type-free, Linear Logic proof nets are Turing complete once
cut-elimination is considered
as computation. We introduce \emph{modal impredicativity} as a
new form of impredicativity causing the complexity of cut-elimination
to be problematic from a complexity point of view.
Modal impredicativity occurs when, during reduction, the conclusion of a
residual of a box $b$
interacts with a node that belongs to the proof net \emph{inside}
another residual of $b$.
Technically speaking, \emph{superlazy reduction} is a new notion of
reduction that allows to control modal impredicativity. More specifically,
superlazy reduction
replicates a box only when all its copies are opened.
This makes the overall cost of reducing a proof net finite and predictable.
Specifically, superlazy reduction applied to any pure proof nets takes
primitive recursive
time. Moreover, any primitive recursive function can be computed by a
pure proof net via superlazy reduction.\par\medskip

\end{abstract}

%% file: introduction.tex
\section{Introduction}
\label{section:Introduction}
Predicativity is a logical concept known from a century, starting from Russel's work. It has
various technical meanings. All of them, however, refer implicitly or explicitly to some form
of aciclicity (see \cite{FefermanPredicativitySurvey} for an excellent survey). Impredicative
definitions or logical rules in a logical system may lead to logical paradoxes. On the other hand,
if logical systems are interpreted as programming languages (via the Curry-Howard correspondence), impredicativity may lead to 
type systems and programming languages with high expressive power.

In this paper, we introduce the notion  \emph{modal impredicativity}. We start from Linear Logic
which gives first-order status to
structural rules (on the logical side) and to duplication and erasure (on the computational side).
The very definition of modal impredicativity refers to \emph{boxes}, i.e., those portions
of proof nets, related to the modal meaning of formulae, that may be duplicated. During cut-elimination,
a duplication occurs when a box interacts with a contraction node, which corresponds to an instance of the structural
rule contraction in a logical derivation. Boxes allow to structure proofs into \emph{layers}: any rule instance has a
\emph{level}, the number of boxes into which it is contained.
Focusing our attention to boxes is the reason why our notion of impredicativity is dubbed as modal.
Specifically, modal impredicativity occurs when, during reduction, the \emph{root} of a residual of a box $b$ interacts with a node that belongs
to the proof net \emph{inside} another residual of $b$.
The paradigmatic example where an interaction of this kind occurs is in the (pure) proof net in Figure~\ref{fig:DeltaDelta},
which encodes the prototypical \emph{non normalizing} lambda-term $(\lambda x.xx)(\lambda x.xx)$. 
Call $b$ the (unique) box in Figure~\ref{fig:DeltaDelta}. After two reduction steps we get two copies
$b'$ and $b''$ of $b$. Two further reduction steps plug the root of $b'$ as premise of the node $X$ belonging to the second copy $b''$ of $b$.
This is a basic form of one-step-long cycle, since the content of $b$ interacts with the root of $b$ itself. Compared to what happens classically,
\emph{self-copying} plays the r\^ole of \emph{self-application} or \emph{self-definition}.
\begin{figure}[!ht]
\begin{center}
  \begin{minipage}[c]{140pt}
    \centering\scalebox{0.6}{\epsfbox{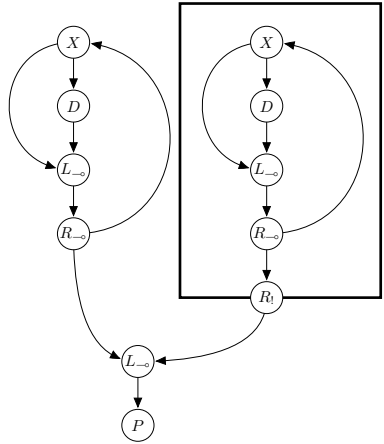}}
  \end{minipage}
\caption{The pure proof net $\Delta\Delta$.}
\label{fig:DeltaDelta}
\end{center}
\end{figure}
Notice that the cycles we are speaking about can have length greater than one. As an example, consider the
proof net corresponding to the lambda term $(\lambda x.\lambda y.xyx)(\lambda x.\lambda y.xyx)(\lambda x.\lambda y.xyx)$: it
includes two boxes $b_1$ and $b_2$ where $b_1$ copies $b_2$ and $b_2$ copies $b_1$.

Our long-term goal is to define proper restrictions on Linear Logic
allowing to control modal impredicativity. This paper is just the first
step towards this goal. What we define here is a new notion of \emph{reduction} for Linear Logic proof nets
which rules out the previously described cyclic phenomenon dynamically, i.e. at the level of the graph-theoretic
rewriting relation which governs proof net reduction.
This way, we break impredicative cycles, while keeping the
freedom of \emph{statically} compose pure proof nets. 
\paragraph{Light Logics and Modal Impredicativity.}
We now recall how the known subsystems of Linear Logic, introduced as characterization of certain complexity classes, work.
Let us call them light logics, for short. Proof-theoretically, they parsimoniously use the contraction rule.
On the computational side, they control the duplication of structure.
Technically, we currently know two ways of controlling the use of duplication.
One is \emph{stratification}. The other one is what we like to call \emph{boundedness}.

\emph{Stratification} is a structural constraint that, at the dynamic level, has the following meaning:
one reduction step at level $n$ can only increase the complexity
of the underlying proof at levels (strictly) higher than $n$. This is achieved by dropping dereliction
and digging as logical rules.
The consequence is the control over the dimension of every single reduct, that reflects on the
overall control of reduction time. The mechanism is implicit in the structural and combinatorial properties of
proofs and is totally independent from its logical soundness. In stratified systems, the level of any node \emph{cannot}
change during reduction. As a consequence, any stratified system, by definition, cannot be modally impredicative because
the nodes inside a box $b$ cannot interact with the root of any copy of $b$.
Elementary Linear Logic, Light Linear Logic~\cite{Girard98ic} and their affine versions use \emph{stratification}.

Concerning boundedness, recall that in ordinary Linear Logic, $!A$ is semantically equivalent to $\bigcup_{n\in\mathbf{N}} \underbrace{(A\otimes \ldots\otimes A)}_{n}$.
\emph{Boundedness} refers to various methodologies that, informally, put
$!A$ in correspondence to a \emph{finite} subset of
$\bigcup_{n\in\mathbf{N}} \underbrace{(A\otimes \ldots\otimes A)}_{n}$.
Computationally, this means the number of copies of each box in a proof can somehow
be statically or dynamically predicted, i.e. bounded.
This way, we automatically get a system that cannot be modally impredicative, since in bounded systems
$!A$ cannot be equal to $!A\otimes A$ (and this principle seems necessary to have
self-copying).
Soft Linear Logic~\cite{Lafont04tcs} and Bounded Linear Logic~\cite{GirardScedrowScott92tcs} use \emph{boundedness}.

\paragraph{Superlazy Reduction and Modal Impredicativity.}
\emph{Superlazy reduction} is a new notion of reduction for
Linear Logic (pure) proof nets. It is specifically designed to control modal impredicativity.
Under superlazy reduction, any box $b$ can interact with a tree of contraction, dereliction, digging and weakening nodes
only if the global result of this interaction somehow reduces the overall complexity
of the proof net, namely when it produces (possibly many) ``open'' copies of $b$.
If this is not the case, reduction is blocked and cannot be performed. This way modal impredicativity is automatically ruled out,
since whenever the \emph{content} of $b$ is copied, $b$ \emph{as a box} is destroyed and no residuals of
$b$ are produced. Technically, this is ensured by prescribing that
reduction can happen only when the box is faced with a \emph{derelicting tree of nodes}, a key notion introduced in
Section~\ref{section:Soundness}.
\paragraph{Superlazy Reduction and Primitive Recursion.}
The calculus we obtain by adopting superlazy reduction over pure proof nets is still
powerful enough to characterize the class \PR\ of primitive recursive functions.
We show the characterization under a standard pattern.
As for \emph{soundness}, we prove that every pure proof net $G$ can be rewritten to its normal form in time
bounded by a primitive recursive function in the size of $G$.
This is remarkable by itself, since the mere fact that superlazy reduction \emph{computes
something} is interesting by itself, considered the strong requirements superlazy reduction must satisfy.
As for \emph{completeness}, every function in \PR\ can be represented as a pure proof net, even under superlazy reduction.
\paragraph{Superlazy Reduction and Expressive Power.}
We here want to make some observations about pure proof nets and superlazy reduction as 
a paradigmatic programming language.
The set of terms we can program with are pure proof nets coming from Linear Logic.
Namely, we can use every lambda-term as a program, $\Delta\Delta$ (Figure~\ref{fig:DeltaDelta}) included.
However, we do not have the standard unconstrained reduction steps, which, by simulating the usual beta-reduction,
allow to embed pure, i.e. untyped, lambda-calculus into pure proof nets.
In particular, the proof net $\Delta\Delta$ is normal in our setting,  the reason being that we can never reach
a situation where the ``amount'' of open copies of the box it contains is known in advance.
\paragraph{Related Work.}
Several authors have used some notion of predicativity as a way to restrict the expressive power of
programming languages or logical systems. We here recall some of them, without
any hope of being exhaustive.
In \cite{LeivantDanner99MSCS}, as a first example, Danner and Leivant presented a variant
of the second order $\lambda$-calculus obtained imposing a restriction on the
second-order quantification. Such a restriction has semantic flavor: all the \emph{types}
have a \emph{rank}, an ordinal number\condinc{; the type $\forall R:\xi.\tau$ is a legal type
where $R$ can be replaced with every other type \emph{of rank at most $\xi$}. Of course the
formula $\forall R:\xi.\tau$ has rank higher than $\xi$, so a form of predicativity is reached.
However, this is not enough to limit the expressive power: the maximum possible rank must be
limited \emph{a priori}.}{ and universal quantification can only be instantiated if the witness
has a proper rank.} This way they get a characterization of primitive recursion. Another, earlier,
example is Leivant's predicative recursion~\cite{Leivant93popl}: if predicativity is imposed
on ordinary primitive recursion (on any word algebra), one gets a characterization of polynomial time
computable functions. Further work shows how other classes can be characterized with similar
tools~\cite{Leivant99apal,Leivant00tcs}. A final example is Simmons' fine analysis of tiering~\cite{Simmons05bsl}. All the cited proposals,
however, share a property which make them fundamentally different from superlazy reduction: predicativity
is enforced through static constraints, i.e., constraints on programs rather than on the underlying
reduction relation. The first author has recently proposed a characterization of primitive recursion
by a fragment of G\"odel's $\mathsf{T}$~\cite{Dallago05lics}.

On the other hand, restricted notions of reduction on Linear Logic proof nets have appeared
in the literature. This includes, for example, Girard's closed reduction~\cite{Girard89lc} or head linear
reduction~\cite{Mascari94tcs}. None of them, however, decreases the expressive power of the logical systems
on top of which they are applied, as superlazy reduction does.

\paragraph{Paper Outline.}
Section~\ref{section:Pure proof nets} recalls pure, i.e., untyped, proof nets and defines what are derelicting trees and
superlazy reduction.
Section~\ref{section:Soundness} proves the primitive recursive soundness of superlazy
reduction on pure proof nets.
Section~\ref{section:Completeness} shows that even under superlazy reduction, pure proof nets remain expressive
enough to represent all the primitive recursive functions.
Section~\ref{section:FurtherDevelopments} presents some further developments on the ideas presented in the paper.
\condinc{}{An extended version of this paper including all proofs is available~\cite{longversion}.}

%% file: proof-nets.tex
\section{Pure Proof Nets}
\label{section:Pure proof nets}
\condinc
{Pure proof nets are graph-like structures built using
the nodes in Figure~\ref{fig:nodes} and
the inductive clauses in Figure~\ref{fig:baseinduction-a}, \ref{fig:baseinduction-b} and~\ref{fig:induction}.}
{Pure proof nets are graph-like structures corresponding to proofs. The set of labels for the nodes is
$\{P$, $C$, $W$, $X$, $R_{\linear}$, $L_{\linear}$, $R_{\otimes}$, $L_{\otimes}$, $R_{!}$, $L_{!}$, $D$, $N \}$.
All nodes, but $P$ and $C$, correspond to the usual proof net labels. We use $P$ and $C$ for the sake of uniformity, getting graphs
without dangling edges.}
Figure~\ref{fig:baseinduction-a} says that a wire is a proof net. Given the two proof nets in~\ref{fig:baseinduction-b} we can build
those in Figure~\ref{fig:induction}.
The inductive rule at the end of Figure~\ref{fig:baseinduction-a} introduces (modal) \emph{boxes}.

\condinc{
\begin{figure}
\begin{center}
  \begin{minipage}[c]{\textwidth}
    \centering\condinc{\scalebox{0.6}}{\scalebox{0.4}}{\epsfbox{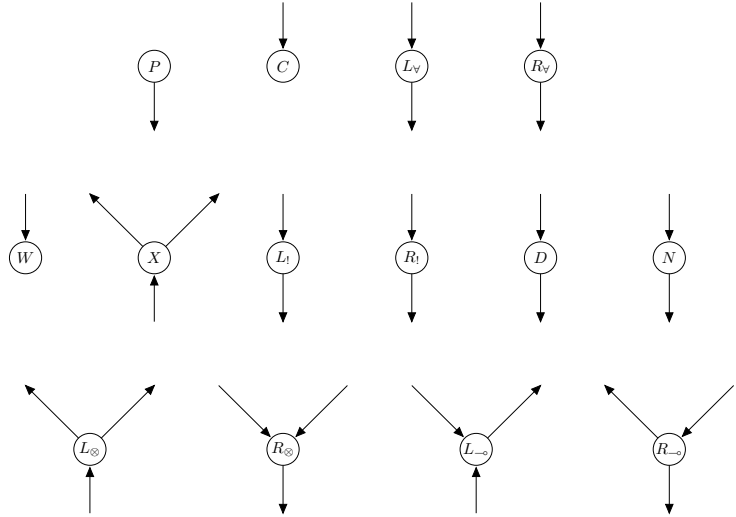}}
  \end{minipage}
\caption{The nodes of the proof nets.}
\label{fig:nodes}
\end{center}
\end{figure}
}{} 
\begin{figure}
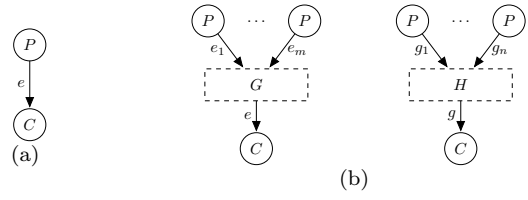

\begin{center}
\subfigure[]{
  \begin{minipage}[c]{13.2pt}
    \centering\condinc{\scalebox{0.6}}{\scalebox{0.6}}{\epsfbox{figure.8}}
  \end{minipage}
\label{fig:baseinduction-a}
}
\hspace{40pt}
\subfigure[]{
  \begin{minipage}[c]{49.8pt}
    \centering\condinc{\scalebox{0.6}}{\scalebox{0.6}}{\epsfbox{figure.2}}
  \end{minipage}
  \hspace{20pt}
  \begin{minipage}[c]{49.8pt}
    \centering\condinc{\scalebox{0.6}}{\scalebox{0.6}}{\epsfbox{figure.58}}
  \end{minipage}
\label{fig:baseinduction-b}
}
\caption{Base cases.}
\label{fig:baseinduction}
\end{center}
\end{figure}
\begin{figure*}
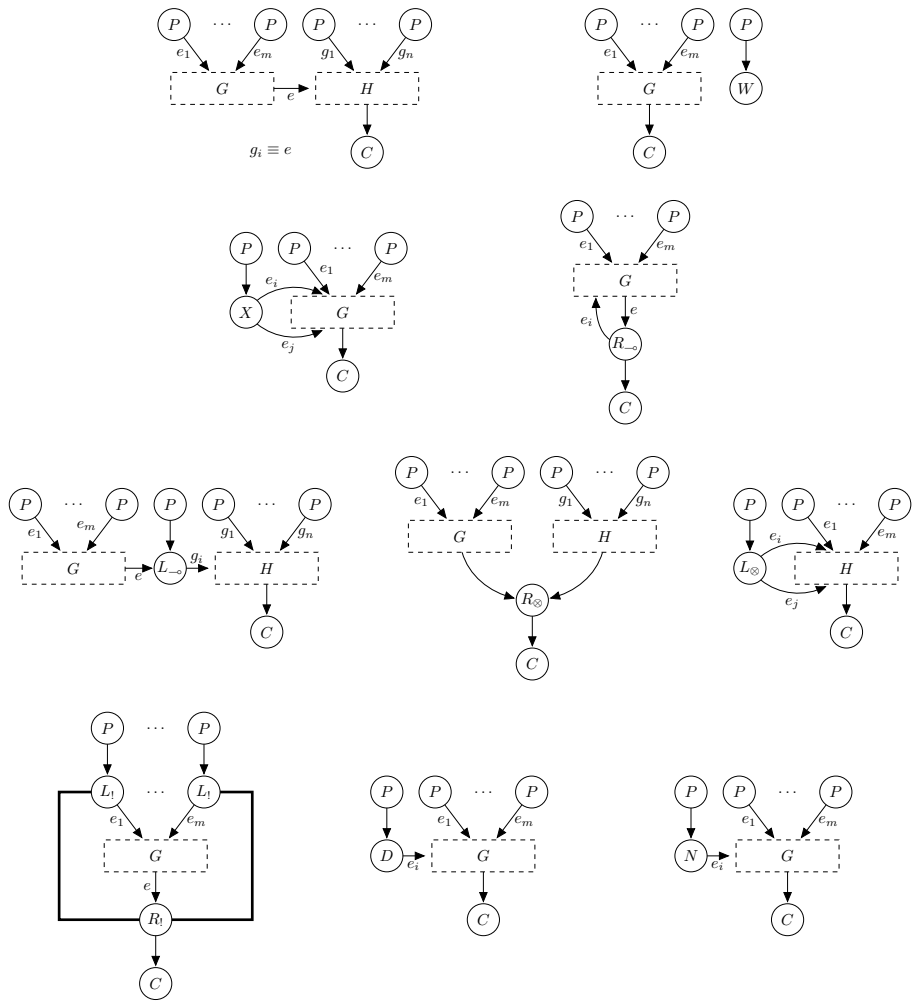

\begin{center}
  \begin{minipage}[c]{103.2pt}
    \centering\condinc{\scalebox{0.6}}{\scalebox{0.6}}{\epsfbox{figure.59}}
  \end{minipage} \hspace{50pt}
  \begin{minipage}[c]{67.2pt}
    \centering\condinc{\scalebox{0.6}}{\scalebox{0.6}}{\epsfbox{figure.10}}
  \end{minipage}\\\vspace{10pt}
  \begin{minipage}[c]{67.2pt}
    \centering\condinc{\scalebox{0.6}}{\scalebox{0.6}}{\epsfbox{figure.25}}
  \end{minipage} \hspace{50pt}
  \begin{minipage}[c]{49.2pt}
    \centering\condinc{\scalebox{0.6}}{\scalebox{0.6}}{\epsfbox{figure.3}}
  \end{minipage}\\\vspace{10pt}
  \begin{minipage}[c]{121.2pt}
    \centering\condinc{\scalebox{0.6}}{\scalebox{0.6}}{\epsfbox{figure.4}}
  \end{minipage}\hspace{20pt}
  \begin{minipage}[c]{103.2pt}
    \centering\condinc{\scalebox{0.6}}{\scalebox{0.6}}{\epsfbox{figure.17}}
  \end{minipage}\hspace{20pt}
  \begin{minipage}[c]{67.2pt}
    \centering\condinc{\scalebox{0.6}}{\scalebox{0.6}}{\epsfbox{figure.16}}
  \end{minipage}\\\vspace{10pt}
  \begin{minipage}[c]{73.2pt}
    \centering\condinc{\scalebox{0.6}}{\scalebox{0.6}}{\epsfbox{figure.7}}
  \end{minipage}\hspace{40pt}
  \begin{minipage}[c]{67.2pt}
    \centering\condinc{\scalebox{0.6}}{\scalebox{0.6}}{\epsfbox{figure.12}}
  \end{minipage} \hspace{40pt}
  \begin{minipage}[c]{67.2pt}
    \centering\condinc{\scalebox{0.6}}{\scalebox{0.6}}{\epsfbox{figure.13}}
  \end{minipage} \hspace{20pt}
\end{center}
\caption{Inductive cases}
\label{fig:induction}
\end{figure*}
Please notice that the proof nets introduced here are slightly different from the usual ones. In particular, there is
not any explicit node playing the r\^ole of the cut rule or of axioms. Moreover, proof net conclusions are partitioned into one proper
conclusion and some premises. This way, proof nets get an intuitionistic flavor which makes the correspondence with lambda-terms more evident\condinc{; this will be useful in Section~\ref{section:Completeness}}{}.

\paragraph{Reduction Rules.}
\condinc{
  The reduction rules for the proof nets are in Figure~\ref{fig:linearrules}, and~\ref{fig:exporules}.
}{
  The reduction rules for the proof nets are the usual ones. We omit the obvious linear rules
  $\longrightarrow_\linear, \longrightarrow_\otimes$, and we just recall the
  \emph{modal} reduction rules in Figure~\ref{fig:exporules}.
} 
Call $\longrightarrow$ the contextual closure of the rewriting steps
$\longrightarrow_\linear, \longrightarrow_\otimes, \longrightarrow_D, \longrightarrow_X, \longrightarrow_W, \longrightarrow_N, \longrightarrow_M$.
The reflexive and transitive closure of $\longrightarrow$ is $\longrightarrow^*$.
\condinc{
\begin{figure}
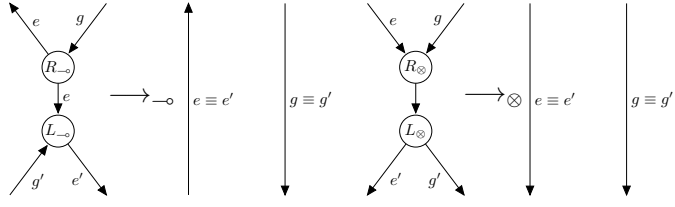

\begin{center}
  \begin{minipage}[c]{34pt}
    \centering\condinc{\scalebox{0.6}}{\scalebox{0.4}}{\epsfbox{figure.14}}
  \end{minipage}
  \begin{minipage}[c]{25pt}
    \centering $\longrightarrow_\linear$
  \end{minipage}
  \begin{minipage}[c]{41.5pt}
    \centering\condinc{\scalebox{0.6}}{\scalebox{0.4}}{\epsfbox{figure.15}}
  \end{minipage}
  \hspace{20pt}
  \begin{minipage}[c]{33pt}
    \centering\condinc{\scalebox{0.6}}{\scalebox{0.4}}{\epsfbox{figure.54}}
  \end{minipage}
  \begin{minipage}[c]{20pt}
    \centering $\longrightarrow_\otimes$
  \end{minipage}
  \begin{minipage}[c]{29pt}
    \centering\condinc{\scalebox{0.6}}{\scalebox{0.4}}{\epsfbox{figure.55}}
  \end{minipage}
  \hspace{20pt}
\caption{Linear graph rewriting rules.}
\label{fig:linearrules}
\end{center}
\end{figure}

}{} 

\begin{figure}
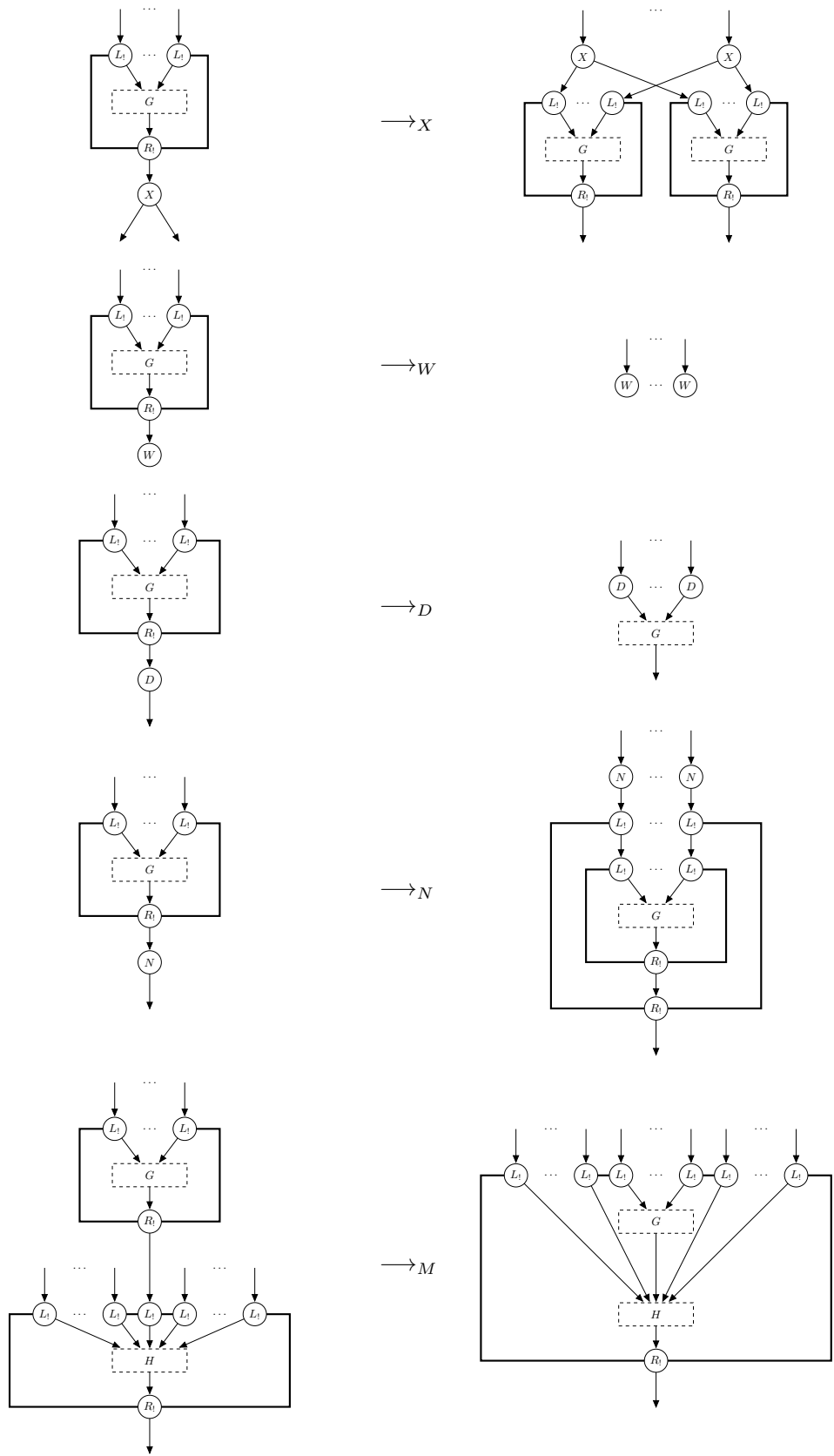

\begin{center}
  \begin{minipage}[c]{.45\textwidth}
    \centering\condinc{\scalebox{0.5}}{\scalebox{0.4}}{\epsfbox{figure.443}}
  \end{minipage}
  \begin{minipage}[c]{20pt}
    \centering $\longrightarrow_X$
  \end{minipage}
  \begin{minipage}[c]{.45\textwidth}
    \centering\condinc{\scalebox{0.5}}{\scalebox{0.4}}{\epsfbox{figure.453}}
  \end{minipage}
\\\vspace{10pt}
  \begin{minipage}[c]{.45\textwidth}
    \centering\condinc{\scalebox{0.5}}{\scalebox{0.4}}{\epsfbox{figure.21}}
  \end{minipage}
  \begin{minipage}[c]{20pt}
    \centering $\longrightarrow_W$
  \end{minipage}
  \begin{minipage}[c]{.45\textwidth}
  \centering\condinc{\scalebox{0.5}}{\scalebox{0.4}}{\epsfbox{figure.22}}
  \end{minipage}
\\\vspace{10pt}
  \begin{minipage}[c]{.45\textwidth}
    \centering\condinc{\scalebox{0.5}}{\scalebox{0.4}}{\epsfbox{figure.40}}
  \end{minipage}
  \begin{minipage}[c]{20pt}
    \centering $\longrightarrow_D$
  \end{minipage}
  \begin{minipage}[c]{.45\textwidth}
    \centering\condinc{\scalebox{0.5}}{\scalebox{0.4}}{\epsfbox{figure.41}}
  \end{minipage}
  \hspace{30pt}
  \begin{minipage}[c]{.45\textwidth}
    \centering\condinc{\scalebox{0.5}}{\scalebox{0.4}}{\epsfbox{figure.42}}
  \end{minipage}
  \begin{minipage}[c]{20pt}
    \centering $\longrightarrow_N$
  \end{minipage}
  \begin{minipage}[c]{.45\textwidth}
    \centering\condinc{\scalebox{0.5}}{\scalebox{0.4}}{\epsfbox{figure.43}}
  \end{minipage}
  \hspace{30pt}\\\vspace{10pt}
  \begin{minipage}[c]{.45\textwidth}
    \centering\condinc{\scalebox{0.5}}{\scalebox{0.4}}{\epsfbox{figure.23}}
  \end{minipage}
  \begin{minipage}[c]{20pt}
    \centering $\longrightarrow_M$
  \end{minipage}
  \begin{minipage}[c]{.45\textwidth}
    \centering\condinc{\scalebox{0.5}}{\scalebox{0.4}}{\epsfbox{figure.24}}
  \end{minipage}
  \hspace{30pt}
\caption{Modal graph rewriting rules.}
\label{fig:exporules}
\end{center}
\end{figure}
The reduction rules $\longrightarrow_X$ and $\longrightarrow_N$ are the only ones somehow increasing
the size of the underlying proof-net: the first one copies a box, while the second one puts a box
inside another box. Superlazy reduction, as we will see shortly, does not simply eliminate those
rules, but rather forces them to be applicable only in certain contexts, i.e., only when those rules
are part of a sequence of modal rewriting rules which have a globally predictable behavior.

\subsection{Superlazy Reduction}
We shall be able to prove a soundness result about the cost of the reduction of the proof nets relatively to
a \emph{superlazy} version of $\longrightarrow$ that requires the notion of \emph{derelicting tree}.

\paragraph{Derelicting trees.}
For every proof net $G$, let us assume that the \emph{cost} of traversing any $X$-node of $G$ is $0$,
any $D$-node is $-1$, any $W$-node is $0$ and any $N$-node is $+1$.
The cost of a path from node $u$ to node $v$ is the sum of the costs of nodes in the path \emph{including $u$ and $v$}.
A \emph{derelicting tree} in $G$ is a subgraph $t$ of $G$ that satisfies the following four conditions:
\begin{varnumlist}
\item
  $t$ only contains nodes $X, D, N, W$; so it must be a tree, and we call $w$ its root;
\item
  the leaves of $t$ are labelled either with $D$ or with $W$;
\item
  for every leaf $v$ labelled with $D$ in $t$, the cost of
  the path from $w$ to $v$ in $t$ is $-1$;
\item
  the cost of any other path in $t$ starting from $w$ is nonnegative.
\end{varnumlist}
Figure~\ref{fig:Derelicting tree-a} shows an example of a derelicting tree. Conditions 1. and 2. are trivially
satisfied.
The cost of $v_1v_2v_4v_7v_{10}$ is $-1$ and the same for $v_1v_3v_6$. Finally, any other path
starting from $v_1$ has nonnegative cost. For example, $v_1v_3v_5v_8$ has cost $1$.

Figure~\ref{fig:Derelicting tree-b} depicts a remarkable instance of derelicting tree: an $n$-\emph{bounded spine} with $n$
occurrences of $X$ nodes that we shall represent as a dashed box with name $nX$.
\begin{figure}[!ht]
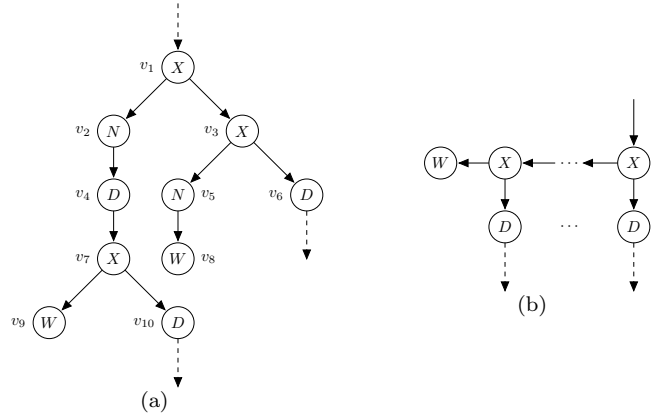

\begin{center}
  \subfigure[]{
  \begin{minipage}[c]{111pt}
    \centering\condinc{\scalebox{0.6}}{\scalebox{0.6}}{\epsfbox{figure.85}}
  \end{minipage}
  \label{fig:Derelicting tree-a}
  }
  \hspace{30pt}
  \subfigure[]{
  \begin{minipage}[c]{85.2pt}
    \centering\condinc{\scalebox{0.6}}{\scalebox{0.6}}{\epsfbox{figure.1001}}
  \end{minipage}
  \label{fig:Derelicting tree-b}
  }
\caption{A generic derelicting tree (a) and a bounded spine (b)}
\end{center}
\end{figure}

\paragraph{Superlazy Reduction Step and Rewriting System on Pure Proof Nets.}
The \emph{superlazy normalization step} is $\rightarrow_{\mathit{XNDW}}$, defined in Figure~\ref{fig:lazycutelim}, $\nabla t$ 
being any derelicting tree.
\begin{figure}[!ht]
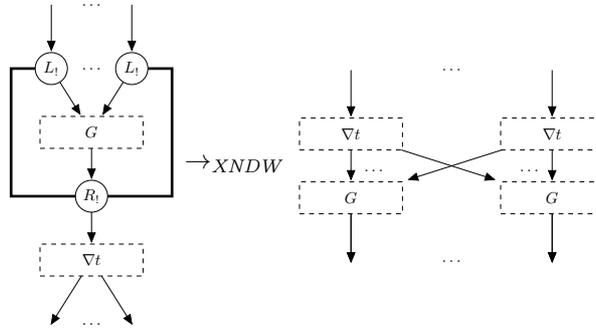

\begin{center}
  \begin{minipage}[c]{61.2pt}
    \centering\condinc{\scalebox{0.6}}{\scalebox{0.6}}{\epsfbox{figure.441}}
  \end{minipage}
  \begin{minipage}[c]{40pt}
    \centering $\rightarrow_{\mathit{XNDW}}$
  \end{minipage}
  \begin{minipage}[c]{114.6pt}
    \centering\condinc{\scalebox{0.6}}{\scalebox{0.6}}{\epsfbox{figure.451}}
  \end{minipage}
\caption{Superlazy cut elimination step}
\label{fig:lazycutelim}
\end{center}
\end{figure}
$\recutel_{\linear}, \recutel_{\otimes}, \recutel_M, \recutel_{\mathit{XNDW}}$
denote the \emph{surface} contextual closure of the rewriting steps
$\longrightarrow_\linear, \longrightarrow_\otimes, \longrightarrow_M, \longrightarrow_{\mathit{XNDW}}$.
``Surface'' means that we never apply a reduction inside a box.
$\recutel$ is the union of $\recutel_{\linear}, \recutel_{\otimes}, \recutel_M, \recutel_{\mathit{XNDW}}$.
The reflexive and transitive closure of $\recutel$ is $\recutel^*$.

Superlazy reduction is a \emph{very} restricted notion of reduction. In particular, it
is almost useless when applied to proof nets obtained from ordinary lambda-terms via the usual,
uniform encoding. In particular:
\begin{varitemize}
\item
  If terms are encoded via the so-called call-by-name encoding (i.e., the one induced
  by Girard's correspondence $\left(A\rightarrow B\right)^\circ\equiv {!A^\circ}\multimap B^\circ$), then
  any redex $(\lambda x.M)N$ where $M$ consists of an application $LP$ and
  $x$ appears free in $P$, cannot be reduced in the corresponding proof net:
  a box would be faced with something different from a derelicting tree.
\item
  On the other hand, if terms are encoded via the call-by-value encoding (i.e. via
  the correspondence $\left(A\rightarrow B\right)^\circ\equiv{!(A^\circ\multimap B^\circ)}$), then
  any redex $(\lambda x.M)N$ where $M$ consists itself of an abstraction $\lambda y.L$
  and $x$ appears free in $L$ cannot be reduced in the corresponding proof net.
\end{varitemize}
Unfortunately, we do not even know any criteria allowing to guarantee that
certain proof nets can be reduced to normal form (w.r.t. ordinary reduction)
by way of superlazy reduction. Moreover, there currently isn't any result
characterizing the class of normal forms w.r.t. superlazy reduction; this is in
contrast, for example, to lambda calculus and call-by-value reduction, where
the cbv normal form of any (closed) term $M$ (if any) is always a value.
This is why proving that proof nets are complete w.r.t.
some given class of functions, under the superlazy reduction, is non-trivial.
Nonetheless, we explicitly prove the completeness of superlazy reduction (see
Section~\ref{section:Completeness}).
The next section shows why both
$\left(A\rightarrow B\right)^\circ\equiv {!A^\circ}\multimap B^\circ$ and  $\left(A\rightarrow B\right)^\circ\equiv{!{(A^\circ\multimap B^\circ)}}$
do not work well here: superlazy reduction of any proof net always terminates in a
time bounded by suitable primitive recursive functions.

Linear Logic with superlazy reduction can be seen as a generalization of the principles of Soft Linear Logic.
Every time a box is replicated, that box is opened.
According to this vision, a derelicting tree with $m$ leaves is similar to a multiplexor node of rank $m$.
However, the structure of \SLL\ gives a further restriction: if $k$
is the rank of the proof net $G$, that is the maximum rank of the multiplexers in $G$, we are sure that every box of $G$ will be copied at most $k$ times. Such a restriction leads to a polytime bound (see \cite{Lafont04tcs}).

%% file: soundness.tex
\section{Soundness}
\label{section:Soundness}
\condinc{ 
	Let's begin with some definitions.
        Given a proof $G$ and any reduction relation $\rightarrow$,
        $\dlength{G}{\rightarrow}$ and $\dsize{G}{\rightarrow}$ denote the maximum
        length of a reduction sequence starting in $G$ (under $\rightarrow$)
        and the maximum size of any reduct of $G$ (under $\rightarrow$), respectively.
        $|G|$ is the size of a proof net $G$.
	If $G\recutel^*F$ and $b^*$ is a box in $F$,
	then $b^*$ is the {merge} of the residuals of one or more boxes
	$b_1',\ldots,b_k'$ in $G$; in this case we say that $b^*$ is the
	\emph{residual* of $\{b_1',\ldots,b_k'\}$}. For shortness we say also that $b^*$ is the residual* of $b_1'$,
	meaning that $b^*$ is the residual* of $b_1'$ plus some other boxes.
	If $b$ is a box in $G$, we denote $\B_G(b)=$ $\{b\}\cup\{d\mid d\mbox{ is a box of $G$ contained in }b\}.$
	A reduction sequence
	$\sigma\equiv G\recutel G'\recutel\ldots\recutel H$ is said to be a
	\emph{$(n,d)$-box reduction} if there are $r\leq n$ boxes $b_1,\ldots,b_r$
	between those at level $0$ in $G$ such that the following conditions hold:
	\begin{varitemize}
		\item
                  The nodes of the proof nets inside every box $b_1,\ldots,b_r$ are at most at depth $d$.
		\item
		  The box $d^*$ involved in any step $\recutel_{\mathit{XNDW}}$ of $\sigma$ is a residual*
		  of some boxes $\{d_1,\ldots,d_k\}$ of $G$ where for every $i\leq k$ there exists $j\leq r$ such that
		  $d_i\in \B_G(b_j)$.
		\item $r$ is the least with the above properties.
	\end{varitemize}

\begin{remark}
A few observations about the above definition:
\begin{varnumlist}
\item
  $G$ may contain more than $n$ boxes at level $0$, or it may contain boxes whose depth is greater than $d$;
\item
  $H$ is \emph{not} necessarily a normal form. It may contain boxes at level $d$, or higher;
\item
  By definition, any node inside the boxes $b_1,\ldots,b_r$ \emph{must have} depth at least $1$. This means that if $\sigma$ is a \emph{$(n,0)$-box reduction} it is, in fact, talking about at most $n$ boxes that contain nodes at level $0$. They cannot exist, so $\sigma$ only uses the linear rewriting steps $\longrightarrow_\linear, \longrightarrow_\otimes$;
\item
  By definition, \emph{every} reduction starting at any net $G$ is a $\left(|G|,\partial(G)\right)$-box reduction because $|G|$ necessarily bounds the number $n$ of boxes, and $\partial(G)$ bounds the value $d$;
\item
	If $b^*$ is a box involved in a $\recutel_{\mathit{XNDW}}$ step and
	it is a residual* of some box $b$ in $G$ at level 0, then
	$b=b_i$ for some $i$ by definition. More important, also the converse holds:
	for each $i$ there exists a residual*	$b_i^*$ of $b_i$ along the reduction
	$G\recutel^*H$	that is involved in a $\recutel_{\mathit{XNDW}}$
	step; that is, $b_i$ will eventually be opened.
\end{varnumlist}
\end{remark}

Now we define a family of functions $\{ f_d:\N\times\N\rightarrow \N\}_{d\in\N}$ such that every $f_d(n,m)$ bounds
\emph{both} the reduction cost \emph{and} the size of
the reducts, of every net $G$ when performing a $(n,d)$-box reductions on it:
	\begin{eqnarray*}
		f_0(n,m)&=&m \\
		f_{d+1}(0,m)&=&m\\
		f_{d+1}(n+1,m)&=& 1+f_{d+1}(n,m)+f_{d}\left(2f_{d+1}(n,m)^2,2f_{d+1}(n,m)^2\right)
	\end{eqnarray*}

By definition, all the functions $f_d$ are primitive recursive.
\begin{lemma}
For every $d,n,m\in\N$, $f_{d+1}(n,m)\geq f_d(n,m)$ and $f_d(n+1,m)\geq f_d(n,m)$.
\end{lemma}
The proof is by induction on $n$. We can prove the following:
\begin{proposition}\label{prop:mainbound}
Let $G$ be any proof net and let $G\recutel^kH$ be a $(n,d)$-box reduction
with $k$ steps. Then $k,|H|\leq f_{d}(n,|G|)$.
\end{proposition}

\begin{proof}
We can proceed by induction on $d$:
\begin{varitemize}
\item
  If $d=0$, thanks to the \emph{aspect} (iii) above,
  the reduction $G\recutel^kH$ only involves multiplicative
  reduction steps. They  strictly reduce the size of the underlying
  net. So, $|H|,k\leq |G|=f_0(n,|G|)$.

\item
  Suppose the thesis holds for $d$ and suppose $G\recutel^kH$
  is a $(n,d+1)$-box reduction. We proceed by another induction,
  this time on $n$:
  \begin{varitemize}
    \item
      If $n=0$, then none of the boxes at level $d+1$ can be reduced.
      As a consequence, $G\recutel^kH$
      only involves multiplicative reduction steps and, again,
      $$
      |H|,k\leq m=f_{d+1}(0,m).
      $$
    \item
      Suppose the thesis holds for $n$ and suppose
      $G\recutel^kH$ is a $(n+1,d+1)$-box reduction sequence.
      By definition, there are $r\leq n+1$ boxes
      $b_1\ldots,b_r$ in $G$ satisfying the definition above.
      If $r=0$, the reduction is performed in linear time, so $|H|\leq|G|\leq f_{d+1}(n+1,m)$.
      Otherwise, we already noted that for each box $b_i$ there is a residue* $b_i^*$
      that will be involved in a $\recutelExp$ step.
      There is clearly one index $t$, where $1\leq t\leq r$,
      such that $b^*_t$ is the last box being copied \emph{among the residuals of the $b_i$'s}
      in the sequence $G\recutel^kH$.
      The reduction under consideration can be decomposed as follows:
      $$
      G\recutel^iF\recutelExp J\recutel^jH
      $$
      where the step $F\recutelExp J$ is the reduction of $b_t^*$.
      Up to the $F$, the reduction sequence can be considered
      as a $(n,d+1)$-box reduction sequence: the witness
      list of boxes is $b_1,\ldots,b_{t-1},b_{t+1},\ldots,b_r$.
      After $J$, on the other hand, the reduction sequence
      can be considered as a $(|J|,d)$-box reduction
      sequence.
      Indeed, by definition of $\sigma$, all the boxes that will be involved in
      a $\recutelExp$ step are residual* of some boxes $\{d_1,\ldots,d_k\}$
      inside $b_1,\ldots,b_r$; but all the boxes $b_i$ have been opened, so
      the residual of their content is less depth.\par
      Applying both inductive hypothesis, we get
      \begin{eqnarray*}
        i,|F|&\leq& f_{d+1}\left(n,m\right)\\
        |J|&\leq& 2|F|^2\qquad\mbox{usual reduction bound}\\
        j,|H|&\leq& f_{d}\left(|J|,|J|\right)
      \end{eqnarray*}
      As a consequence:
      \begin{eqnarray*}
				|J|  &\leq & 2f_{d+1}(n,m)^2\\
        k       &=    & i+1+r\leq f_{d+1}(n,m)+1+f_{d}(|J|,|J|)\\
                &\leq & f_{d+1}(n,m)+1+f_d\left(2f_{d+1}(n,m)^2,2f_{d+1}(n,m)^2\right)\\
        |H|&\leq & f_d(|J|,|J|)< f_{d+1}(n,m)+1+f_d(|J|,|J|)\\
                &\leq & f_{d+1}(n,m)+1+f_d\left(2f_{d+1}(n,m)^2,2f_{d+1}(n,m)^2\right)
      \end{eqnarray*}
  \end{varitemize}
\end{varitemize}
This concludes the proof.
\end{proof}

\begin{corollary}\label{cor:mainbound}
	For every $n\in\N$ there is a primitive recursive function
	$g_n:\N\rightarrow\N$ such that for every proof net $G$,
	$\dlength{G}{\recutel},\dsize{G}{\recutel}\leq g_{\partial(G)}(|G|)$.
\end{corollary}
\begin{proof}
Remember that every reduction sequence
from $\Pi$ is a $(|\Pi|,\partial(\Pi))$-reduction sequence.
So, we can use Proposition~\ref{prop:mainbound} and choose $g_d(n)=f_d(n,n)$.
\end{proof}

} 
{ 

We prove the result starting with a restriction $\hookrightarrow$ on the relation $\recutel$.
$\hookrightarrow$ is simply the union of $\recutel_{\linear}$, $\recutel_{\otimes}$, $\recutel_{M}$
and $\recutel_{\mathit{XNDW}}$, where any $\recutel_{M}$ can only be applied to nets which only
contain $\recutel_{M}$ redexes. In other words, $\recutel_{M}$ is postponed as much as possible.
This makes our arguments simpler without loss of generality, as we now show.
$\hookrightarrow^k$ will denote a reduction with $k$ steps of $\hookrightarrow$.
Given a proof $G$ and any reduction relation $\rightarrow$,
$\dlength{G}{\rightarrow}$ and $\dsize{G}{\rightarrow}$ denote the maximum
length of a reduction sequence starting in $G$ (under $\rightarrow$)
and the maximum size of any reduct of $G$ (under $\rightarrow$), respectively.
We note $|G|$ the size of a proof net $G$.

\begin{lemma}\label{lemma:standard}
	For every proof $G$,
	$\dlength{G}{\recutel}=\dlength{G}{\hookrightarrow}$ and
	$\dsize{G}{\recutel}=\dsize{G}{\hookrightarrow}$.
\end{lemma}

\begin{proof}
Whenever $G\recutel_{M}F\recutel_xH$ and $x\neq M$,
there are $F_1,\ldots,F_n$
(where $n\geq 1$) such that
$
G\recutel_{x_1}F_1\recutel_{x_2}\cdots\recutel_{x_n}F_n\recutel_{x_{n+1}}H,
$
and $x_{i+1}=M$ whenever $x_i=M$. For example,
if $G\recutel_{M}F\recutel_{XNDW}H$
and the box copied in the second step is exactly the one created by
the first step, then clearly
$$
G\recutel_{XNDW}J\recutel_{XNDW}H.
$$
As a consequence, for any sequence $G_1\recutel\cdots\recutel G_n$ there
is another sequence $F_1\hookrightarrow\cdots\hookrightarrow F_m$ such
that $G_1=F_1$, $G_n=F_m$ and $m\geq n$. This proves the first claim.
Now, observe that for any $1\leq i\leq n$ there is $j$ such
that $|F_j|\geq |G_i|$: a simple case analysis suffices.
This concludes the proof.
\end{proof}

The following definition is the main ingredient since it allows to structure the proof of soundness inductively.
A reduction sequence
$\sigma\equiv G\recutel_{x_1}F\recutel_{x_2}\ldots\recutel_{x_k}H$ is said to be a \emph{$(n,d)$-box reduction}
when $x_i\neq M$ for every $1\leq i\leq k$ and
there are $r\leq n$ boxes $b_1,\ldots,b_r$ between those at level $0$ in $G$ such that
the following conditions hold
(let $J_1,\ldots,J_r$ be the proof nets inside $b_1,\ldots,b_r$, respectively):
\begin{varitemize}
\item
  $\partial(J_1),\ldots,\partial(J_r)< d$.
\item
  The box involved in any step $\recutel_{\mathit{XNDW}}$ of $\sigma$ is either a residual of
  $b_1,\ldots,b_r$ or 
  a residual of a box appearing in one of $J_1,\ldots,J_r$.
\end{varitemize}
\begin{remark}\label{rmrk:ndredution}
A few observations about the above definition:
\begin{varnumlist}
\item
  $G$ may contain more than $n$ boxes at level $0$, or it may contain boxes whose depth is greater than $d$;
\item
  $H$ is \emph{not} necessarily a normal form. It may contain boxes at level $d$, or higher;
\item
  By definition, any node inside the boxes $b_1,\ldots,b_r$ \emph{must have} depth at least $1$. This means that if $\sigma$ is a \emph{$(n,0)$-box reduction} it is, in fact, talking about boxes that contain proof nets with negative depth. They cannot exist,
so $\sigma$ only uses the linear rewriting steps $\longrightarrow_\linear, \longrightarrow_\otimes$;
\item
  By definition, \emph{every} reduction starting at any net $G$ is a $\left(|G|,\partial(G)\right)$-box reduction because $|G|$ necessarily bounds the number $n$ of boxes, and $\partial(G)$ bounds the value $d$;
\item
  The definition is a good one \emph{only} because we are assuming to
work without the rule $\longrightarrow_M$. Otherwise (residuals of) boxes coming from
$b_1,\ldots,b_r$ could be ``merged'' with (residuals of) boxes in $G$ but not in
the list.
\end{varnumlist}
\end{remark}

Now we define a family of functions $\{ f_d:\N\times\N\rightarrow \N\}_{d\in\N}$ such that every $f_d(n,m)$ bounds
\emph{both} the reduction cost \emph{and} the size of
the reducts, of every net $G$ when performing a $(n,d)$-box reduction on it:
\begin{eqnarray*}
   f_0(n,m)&=&m \\
   f_{d+1}(0,m)&=&m\\
   f_{d+1}(n+1,m)&=&1+f_{d+1}(n,m)+f_{d}((2m+1)f_{d+1}(n,m),(2m+1)f_{d+1}(n,m))
\end{eqnarray*}
By definition, all the functions $f_d$ are primitive recursive.
\begin{lemma}
For every $d,n,m\in\N$, $f_{d+1}(n,m)\geq f_d(n,m)$ and $f_d(n+1,m)\geq f_d(n,m)$.
\end{lemma}
The proof is by induction on $n$. We can prove the following:
\begin{proposition}\label{prop:mainbound}
Let $G$ be any proof net and let $G\hookrightarrow^kH$ be a $(n,d)$-box reduction
sequence with $k$ steps. Then
$k,|H|\leq f_{d}(n,|G|)$.
\end{proposition}
\begin{proof}
We write $m=|G|$. We can proceed by induction on $d$:
\begin{varitemize}
\item
  If $d=0$, thanks to Remark~\ref{rmrk:ndredution}.3 above,
  the reduction $G\hookrightarrow^kH$ only involves multiplicative
  reduction steps. They  strictly reduce the size of the underlying
  net. So, $|H|,k\leq m=f_0(n,m)$.

\item
  Suppose the thesis holds for $d$ and suppose $G\hookrightarrow^kH$
  is a $(n,d+1)$-box reduction. We proceed by another induction,
  this time on $n$:
  \begin{varitemize}
    \item
      If $n=0$, then none of the boxes at level $d+1$ can be reduced.
      As a consequence, $G\hookrightarrow^kH$
      only involves multiplicative reduction steps and, again,
      $$
      |H|,k\leq m=f_{d+1}(0,m).
      $$
    \item
      Suppose the thesis holds for $n$ and suppose
      $G\hookrightarrow^kH$ is a $n+1$-box reduction sequence at
      level $d+1$. By definition, there are $r\leq n+1$ boxes
      $b_1\ldots,b_r$ in $G$ satisfying the definition above.
      There is clearly one index $t$, where $1\leq t\leq r$
      such that $b_t$ is the last box being copied (among $b_1,\ldots,b_r$)
      in the sequence $G\hookrightarrow^kH$. Up to the point where
      $b_t$ is copied, the reduction sequence can be considered
      as a $(n,d+1)$-box reduction sequence: the witness
      list of box is exactly $b_1,\ldots,b_{t-1},b_{t+1},\ldots,b_r$.
      After $b_t$ is copied, on the other hand, the reduction sequence
      under consideration can be considered as a $\left(|J|,d\right)$-box reduction sequence,
      since all the boxes that will be copied are
      part of the residual of the content of boxes in the list $b_1,\ldots,b_r$.
      The reduction under consideration
      can be decomposed as follows
      $$
      G\hookrightarrow^iF\recutel_{XNDW}J\hookrightarrow^jH
      $$
      where $G\hookrightarrow^iF$ is an $(n,d+1)$-box reduction
      sequence, the step
      $F\recutel_{XNDW}J$ involves \emph{exactly} $b_t$
      and $J\hookrightarrow^jH$ is a $\left(|J|,d\right)$-box reduction.
      Now, observe that $|b_t|\leq m$. In $|J|$ we will find
      at most $|F|$ copies of $b_t$ and at most $|b_t|$ copies of the
      underlying derelicting tree (which can itself contain at most
      $|F|$ nodes). Applying both inductive hypothesis, we get
      \begin{align*}
        i,|F|&\leq f_{d+1}(n,m)\\
        |J|&\leq 2m f_{d+1}(n,m)+|F|\\
        j,|H|&\leq f_{d}(|J|,|J|)
      \end{align*}
      As a consequence:
      \begin{align*}
        k&=i+1+j\leq f_{d+1}(n,m)+1+f_{d}(|J|,|J|)\\
        &\leq f_{d+1}(n,m)+1+f_d((2m+1)f_{d+1}(n,m),(2m+1)f_{d+1}(n,m))\\
        |H|&\leq f_d(|J|,|J|)\leq f_{d+1}(n,m)+1+f_d(|J|,|J|)\\
        &\leq f_{d+1}(n,m)+1+f_d((2m+1)f_{d+1}(n,m),(2m+1)f_{d+1}(n,m))\\
      \end{align*}
  \end{varitemize}
\end{varitemize}
This concludes the proof.
\end{proof}
\begin{corollary}[Soundness]\label{cor:mainbound}
For every $n\in\N$ there is a primitive recursive function
	$g_n:\N\rightarrow\N$ such that for every proof net $G$,
	$\dlength{G}{\recutel},\dsize{G}{\recutel}\leq g_{\partial(G)}(|G|)$.
\end{corollary}
\begin{proof}
By Lemma~\ref{lemma:standard}, we can bound
$\dlength{G}{\hookrightarrow}$ and $\dsize{G}{\hookrightarrow}$
rather than $\dlength{G}{\recutel}$ and $\dsize{G}{\recutel}$.
Now, observe that any reduction $G\hookrightarrow^kH$ is a $\left(|G|,\partial(G)\right)$-box
reduction, followed by at most
a linear number of $M$-reduction steps, which anyway shrinks the size of
the underlying proof net.
The thesis follows from Proposition~\ref{prop:mainbound}.
\end{proof}

} 

%% file: completeness.tex
\section{Completeness}
\label{section:Completeness}
The goal is to show the existence of an embedding of any primitive
recursive function $\funone$ into a pure proof net $G_\funone$ such that
$G_f$ simulates $\funone$ via superlazy reduction. The existence of such an embedding is what
we mean by \emph{completeness}.
\par\medskip
\condinc{
\subsection{Recalling the Primitive Recursive Functions}
\label{subsection:Recalling the Primitive Recursive Functions}
The primitive recursive functions (\PR) are functions from tuples of natural numbers to
a natural number. \PR\ is the least set that both contains the \textit{basic} primitive
recursive functions and which is closed under a finite number of applications of
\textit{composition} and \textit{primitive recursion}.

\paragraph{Basic Primitive Functions.}
The basic primitive functions are the 0-ary \textit{constant} function $0$,
the $1$-ary \textit{successor} function $s:\N\rightarrow\N$, and the
$m$-ary \textit{projection} function $\pi_i^m:\N^m\rightarrow\N$,
for every $m\geq 1$ and $1\leq i\leq m$. For every natural number
$n$, the natural number $s(n)$ is simply $n+1$, while
$\pi_i^m(n_1,\ldots,n_m)$ is the $i$-th argument $n_i$,
for every $n_1,\ldots,n_m\in\N$.

\paragraph{Composition.}
Let $\funone$ be an $n$-ary
\PR\ function and let $\funtwo_1,\ldots,\funtwo_n$ be $m$-ary \PR\ functions.
The \textit{composition} of $\funone$ with $\funtwo_1,\ldots, \funtwo_n$
is the $m$-ary function defined as follows:
$$
\prcom{\funone}{\funtwo_1,\ldots,\funtwo_n}(n_1,\ldots,n_m)=
\funone(\funtwo_1(n_1,\ldots,n_m),\ldots,\funtwo_n(n_1,\ldots,n_m)).
$$

\paragraph{Primitive Recursion.}
Let $\funone$ be an $m$-ary \PR\ function and let $\funtwo$ be an ($m+2$)-ary \PR\ function.
The function defined by \textit{primitive recursion} on $\funone$ and $\funtwo$ is
the ($m+1$)-ary function defined as follows:
\begin{equation} \label{eqn:prec}
	\begin{split}
		\prrec{\funone}{\funtwo}(0,n_1,\ldots, n_m)&=\funone(n_1,\ldots,n_m) \\
		\prrec{\funone}{\funtwo}(n+1,n_1,\ldots,n_m)&=
		\funtwo(n,\prrec{\funone}{\funtwo}(n,n_1,\ldots,n_m),n_1,\ldots,n_m)
	\end{split}
\end{equation}

\subsection{Representation of Integer Functions}
We represent (tuples) of natural numbers and functions from naturals to natural as proof nets.\par

\paragraph{Clusters.}
Figure~\ref{fig:Clusters} introduces clusters of, at least, one instance of the nodes $L_\linear, L_\otimes, R_\otimes$.
\par

\begin{figure}[!ht]
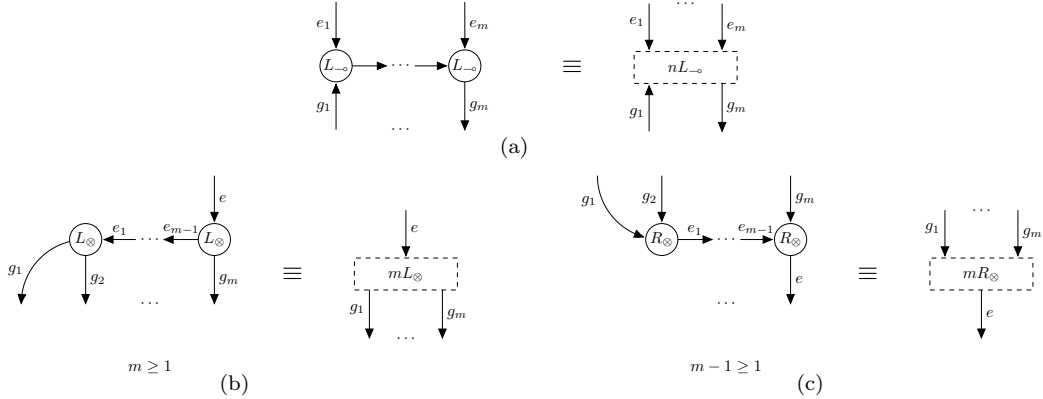

\begin{center}
\subfigure[]{
  \begin{minipage}[c]{110pt}
    \centering\scalebox{0.6}{\epsfbox{figure.102}}
  \end{minipage}
  \begin{minipage}[c]{10pt}
    \centering $\equiv$
  \end{minipage}
  \begin{minipage}[c]{70pt}
    \centering\scalebox{0.6}{\epsfbox{figure.103}}
  \end{minipage}
\label{subf:cluster-L-lin}
}

\begin{tabular}{rl}
\subfigure[]{
  \begin{minipage}[c]{110pt}
    \centering\scalebox{0.6}{\epsfbox{figure.106}}
  \end{minipage}
  \begin{minipage}[c]{10pt}
    \centering $\equiv$
  \end{minipage}
  \begin{minipage}[c]{70pt}
    \centering\scalebox{0.6}{\epsfbox{figure.107}}
  \end{minipage}
\label{subf:cluster-L-oti}
}
&
\subfigure[]{
  \begin{minipage}[c]{110pt}
    \centering\scalebox{0.6}{\epsfbox{figure.108}}
  \end{minipage}
  \begin{minipage}[c]{10pt}
    \centering $\equiv$
  \end{minipage}
  \begin{minipage}[c]{70pt}
    \centering\scalebox{0.6}{\epsfbox{figure.109}}
  \end{minipage}
\label{subf:cluster-R-lin}
}
\end{tabular}

\caption{Clusters of nodes}
\label{fig:Clusters}
\end{center}
\end{figure}

\paragraph{Numerals.}
Figure~\ref{fig:ChurchNumeralAsPN} introduces the closed proof nets $\overline{0}$ and
$\overline{n}$, with $n\geq 1$. $nX$ is the $n$-bounded spine in Figure~\ref{fig:Derelicting tree-b}.
A basic observation to prove the completeness is that every $\overline{n}$ contains a
$n$-bounded spine.

\begin{figure}[ht]
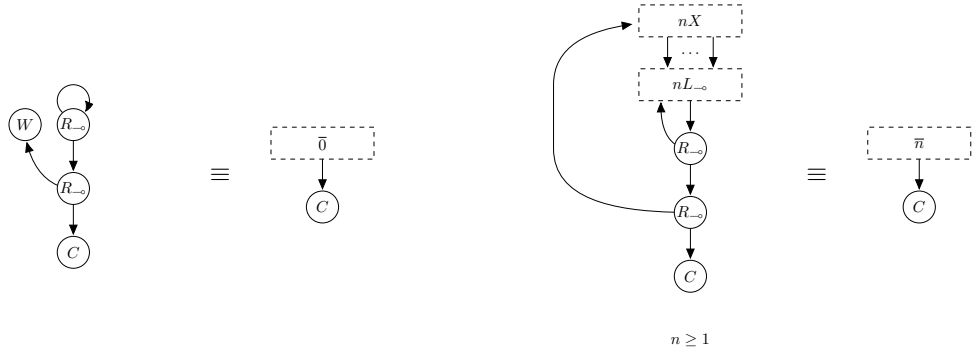

\begin{center}
  \begin{minipage}[c]{110pt}
    \centering\scalebox{0.6}{\epsfbox{figure.1041}}
  \end{minipage}
  \begin{minipage}[c]{10pt}
    \centering $\equiv$
  \end{minipage}
  \begin{minipage}[c]{60pt}
    \centering\scalebox{0.6}{\epsfbox{figure.1051}}
  \end{minipage}
\hspace{30pt}
  \begin{minipage}[c]{110pt}
    \centering\scalebox{0.6}{\epsfbox{figure.104}}
  \end{minipage}
  \begin{minipage}[c]{10pt}
    \centering $\equiv$
  \end{minipage}
  \begin{minipage}[c]{60pt}
    \centering\scalebox{0.6}{\epsfbox{figure.105}}
  \end{minipage}
\caption{Church numerals as proof nets}
\label{fig:ChurchNumeralAsPN}
\end{center}
\end{figure}

\paragraph{Tuples of numerals.}
The proof nets $\tuple{\overline{n_1},\ldots,\overline{n_m}}$ are obtained as follows.
If $m=1$, then $\tuple{\overline{n_1}}$ coincides to the net $\overline{n_1}$.
Otherwise, if $m>1$, then $\tuple{\overline{n_1},\ldots,\overline{n_m}}$ is obtained by composing the $m$ conclusions of every proof net $\overline{n_1},\ldots,\overline{n_m}$ with one of the premises of $mR_\otimes$.

\paragraph{Integer functions.}
Let $G$ be a proof net with a single premise and a single conclusion.
We write $G\smile\tuple{\overline{n_1},\ldots,\overline{n_m}}$ to mean the closed proof net with a single conclusion obtained by plugging the conclusion of the net $\tuple{\overline{n_1},\ldots,\overline{n_m}}$ into the assumption of $G$.\par
For every function $\funone:\N^m\rightarrow\N$, with arity $m\geq 0$, we shall say that a proof net $\prtopn{\funone}$ with one premise \emph{represents} $\funone$ iff
$\prtopn{\funone}\smile\tuple{\overline{n_1},\ldots,\overline{n_m}}$ reduces
to $\overline{\funone(n_1,\ldots,n_m)}$,
for every $n_1,\ldots,n_m\in\N$.

\subsection{The Completeness theorem}

Now we can state formally the Completeness theorem:

\begin{theorem}[Completeness]
\label{theorem:Completeness}
Every $\funone$ in \PR\ is represented by a proof net $\prtopn{\funone}$.
\end{theorem}
To prove it we define how to map every $\funone\in\PR$ into a corresponding $G_\funone$. Then we show that superlazy reduction
allows to rewrite any $\prtopn{\funone}\smile\tuple{\overline{n_1},\ldots,\overline{n_m}}$
to $\overline{\funone(n_1,\ldots,n_m)}$.
To map a function into a net requires to program with the proof nets as we do here below.

\paragraph{The Successor.}
Figure~\ref{fig:Successor} defines the proof net $G_\Succ$ that represents the successor function.
$\prtopn{\Succ}\smile\tuple{\overline{n}}$ yields $\overline{n+1}$ because $\prtopn{\Succ}$ extends the bounded spine $nX$ and $nL_{\multimap}$ of $\overline{n}$ into $(n+1)X$ and $(n+1)L_{\multimap}$, the core of $\overline{n+1}$.

\begin{figure}[ht]
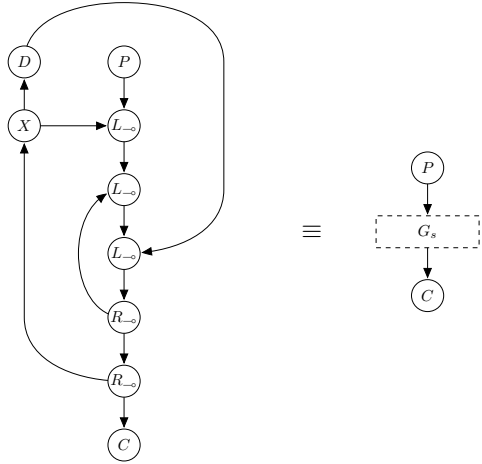

\begin{center}
  \begin{minipage}[c]{130pt}
    \centering\scalebox{0.6}{\epsfbox{figure.122}}
  \end{minipage}
  \begin{minipage}[c]{10pt}
    \centering $\equiv$
  \end{minipage}
  \begin{minipage}[c]{70pt}
    \centering\scalebox{0.6}{\epsfbox{figure.123}}
  \end{minipage}
\caption{The proof net $G_\Succ$ (successor)}
\label{fig:Successor}
\end{center}
\end{figure}

\paragraph{Unity and Erasure.}
	Figure~\ref{fig:Unity} introduces the closed proof net Unity that
	we identify as $\star$.
	Figure~\ref{fig:Erasure} introduces the proof net $\erase$ with two occurrences of $\star$ in it.
$\erase\smile\tuple{\overline n}$ reduces to $\star$,
because
(i) the conclusion of the box of the topmost $\star$ is plugged into the root of the bounded spine of $\tuple{\overline{n}}\equiv\overline{n}$, and (ii) by means of $\recutelExp$, we get $n$ instances of the identity, ``iteratively'' applied to the lower occurrence of $\star$, which, after some reduction transform to $\star$.

\begin{figure}[ht]
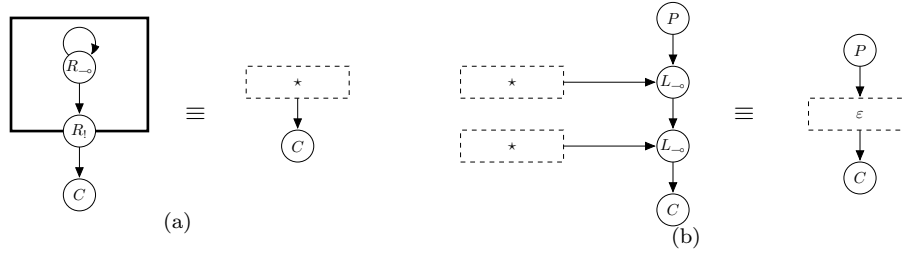

\begin{center}
\begin{tabular}{rl}
  \subfigure[]{
  \begin{minipage}[c]{70pt}
    \centering\scalebox{0.6}{\epsfbox{figure.114}}
  \end{minipage}
  \begin{minipage}[c]{10pt}
    \centering $\equiv$
  \end{minipage}
  \begin{minipage}[c]{60pt}
    \centering\scalebox{0.6}{\epsfbox{figure.115}}
  \end{minipage}
  \label{fig:Unity}
  }
&
  \subfigure[]{
  \begin{minipage}[c]{110pt}
    \centering\scalebox{0.6}{\epsfbox{figure.112}}
  \end{minipage}
  \begin{minipage}[c]{10pt}
    \centering $\equiv$
  \end{minipage}
  \begin{minipage}[c]{70pt}
    \centering\scalebox{0.6}{\epsfbox{figure.113}}
  \end{minipage}
  \label{fig:Erasure}
  }
\end{tabular}
\caption{The unity proof net ($\star$) and the erasure proof net ($\erase$)}
\label{fig:Unity-and-Erasure}
\end{center}
\end{figure}

\paragraph{Diagonal.}
Figure~\ref{fig:Diagonal} introduces the proof net $\diagonal$ that duplicates a given unary string.
$\diagonal\smile\tuple{\overline n}$ reduces to $\tuple{\overline n,\overline n}$ because the conclusion of the box in the definition of $\diagonal$ is plugged into the root of the bounded spine in $\tuple{\overline{n}}\equiv\overline{n}$.
\par

There is an obvious generalization $m\diagonal$ of $\diagonal$, with $m\geq 2$, necessary to represent the composition of functions in \PR. $m\diagonal$ has a single premise and a single conclusion, and it multiplies a given argument $m$ times.
\par

\begin{figure}[ht]
\begin{center}
  \begin{minipage}[c]{220pt}
    \centering\scalebox{0.6}{\epsfbox{figure.110}}
  \end{minipage}
  \begin{minipage}[c]{10pt}
    \centering $\equiv$
  \end{minipage}
  \begin{minipage}[c]{60pt}
    \centering\scalebox{0.6}{\epsfbox{figure.111}}
  \end{minipage}
\caption{The proof net $\diagonal$ (duplicator)}
\label{fig:Diagonal}
\end{center}
\end{figure}

\paragraph{Projections.}
Figure~\ref{fig:Projection} defines the proof net $G_{\proj^{m}_{i}}$ that implements
the basic primitive recursive function $\pi^{m}_{i}$, for any $m\geq 1$, with $1\leq i\leq m$.
The proof net has the expected behavior thanks  to the correct behavior of $\erase$.

\begin{figure}[ht]
\begin{center}
  \begin{minipage}[c]{130pt}
    \centering\scalebox{0.6}{\epsfbox{figure.120}}
  \end{minipage}
  \begin{minipage}[c]{10pt}
    \centering $\equiv$
  \end{minipage}
  \begin{minipage}[c]{70pt}
    \centering\scalebox{0.6}{\epsfbox{figure.121}}
  \end{minipage}
\caption{The proof net $G_{\proj^{m}_{i}}$}
\label{fig:Projection}
\end{center}
\end{figure}

\paragraph{Composition.}
Figure~\ref{fig:Composition} introduces the proof net $G_{\prcom{\funone}{\funtwo_1,\ldots,\funtwo_n}}$ that represents
the function $h:\N^m\rightarrow\N$ obtained by composition from $\funone:\N^n\rightarrow\N$
and $\funtwo_1,\ldots,\funtwo_n:\N^m\rightarrow\N$. Of course the proof nets $G_\funone,G_\funtwo$ must represent
the functions $\funone$, $\funtwo$.

\begin{figure}[ht]
\begin{center}
  \begin{minipage}[c]{140pt}
    \centering\scalebox{0.6}{\epsfbox{figure.128}}
  \end{minipage}
  \begin{minipage}[c]{10pt}
    \centering $\equiv$
  \end{minipage}
  \begin{minipage}[c]{60pt}
    \centering\scalebox{0.6}{\epsfbox{figure.129}}
  \end{minipage}
\caption{The proof net $\prtopn{\prcom{\funone}{\funtwo_1,\ldots,\funtwo_n}}$}
\label{fig:Composition}
\end{center}
\end{figure}

\paragraph{Primitive Recursion.}
We assume that the function $\funone:\N^m\rightarrow\N$ can be represented by the proof
net $G_\funone$, and that $\funtwo:\N^{m+2}\rightarrow\N$ can be represented by the proof net $G_\funtwo$.
Figure~\ref{fig:PrimRec} introduces the proof net $G_{\prrec{\funone}{\funtwo}}$.
We shall show that $G_{\prrec{\funone}{\funtwo}}$ represents the
$(m+1)$-ary primitive recursion $\prrec{\funone}{\funtwo}$, as defined by equations (\ref{eqn:prec}).
$G_{\prrec{\funone}{\funtwo}}$ is built using two more proof nets $G_1$ and $G_2$.
\begin{figure}[ht]
\begin{center}
  \begin{minipage}[c]{130pt} 
    \centering\scalebox{0.6}{\epsfbox{figure.126}}
  \end{minipage}
  \begin{minipage}[c]{10pt}
    \centering $\equiv$
  \end{minipage}
  \begin{minipage}[c]{50pt}  
    \centering\scalebox{0.6}{\epsfbox{figure.127}}
  \end{minipage}
  \qquad
	\begin{minipage}[c]{110pt} 
    \centering\scalebox{0.6}{\epsfbox{figure.130}}
  \end{minipage}
  \begin{minipage}[c]{10pt}
    \centering $\equiv$
  \end{minipage}
  \begin{minipage}[c]{50pt} 
    \centering\scalebox{0.6}{\epsfbox{figure.131}}
  \end{minipage}
\caption{The proof nets $\prtopn{\prrec{f}{g}}$ and $G_1$}
\label{fig:PrimRec}
\end{center}
\end{figure}
The base function $G_1$ of $\prtopn{\prrec{f}{g}}$ is in Figure~\ref{fig:PrimRec}.
$G_1$ is a closed net that, after an application, yields a function $f'$ that, taking a tuple $\overrightarrow{n}$ of $m$ integers, gives a pair of naturals. A little bit more formally, $f'\left(\overrightarrow{n}\right)=\tuple{0, \funone\left(\overrightarrow{n}\right)}$.
Instead, the net $G_2$ in Figure~\ref{fig:SPrimRec} is the \emph{iterable  step function} of $\prtopn{\prrec{f}{g}}$.
Iterable means that $G_2$ can be used as the first argument
of a Church numeral. $G_2$ maps a net that behaves as a function with an ideal type $\N^m\rightarrow\N^2$ to a net with the same type.
More precisely, $g'\left(\tuple{n_0,h},\overrightarrow n\right)=\tuple{n_0+1,g\left(n_0,h\left(\overrightarrow n\right),\overrightarrow n\right)}$.
In particular, the iteration starts from the result of $G_1$, morally having the right type.

\begin{figure}[ht]
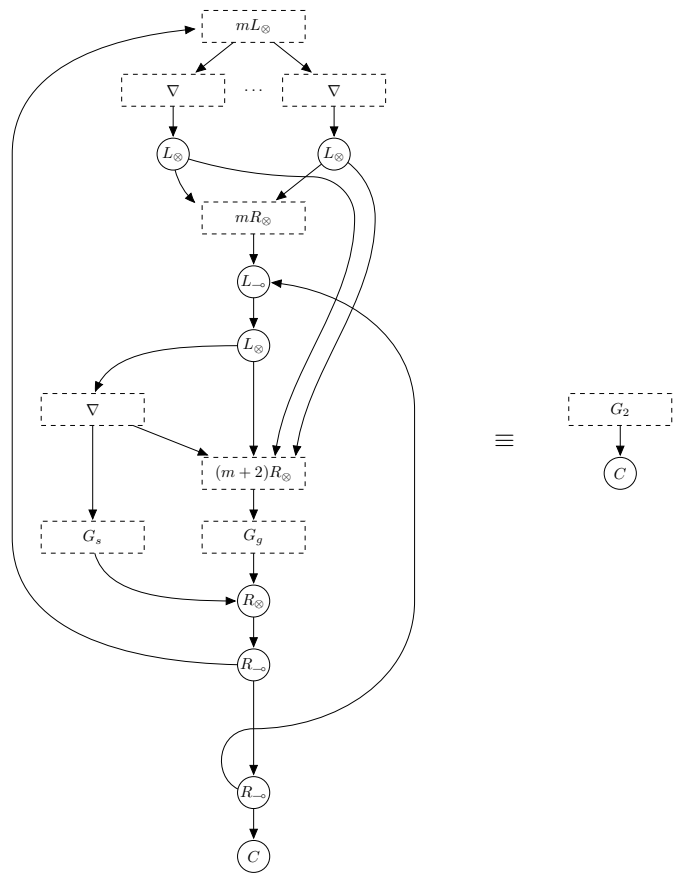

\begin{center}
  \begin{minipage}[c]{200pt}
    \centering\scalebox{0.6}{\epsfbox{figure.124}}
  \end{minipage}
  \begin{minipage}[c]{10pt}
    \centering $\equiv$
  \end{minipage}
  \begin{minipage}[c]{70pt}
    \centering\scalebox{0.6}{\epsfbox{figure.125}}
  \end{minipage}
\caption{The proof net $G_2$}
\label{fig:SPrimRec}
\end{center}
\end{figure}

Now, we want to prove that $G_{\prrec{\funone}{\funtwo}}\smile\tuple{\overline{n_0},\overline{n_1},\ldots,\overline{n_m}}$ reduces to the expected result $\overline{\prrec{f}{g}\left(n_0,\overrightarrow n\right)}$. The keypoint is the application of $G_2$, inside $\prtopn{\prrec{\funone}{\funtwo}}$, as an argument of $n_0$. If $n_0=0$, then the whole box around $G_2$ is erased. This returns the identity. But, for every $n_0>0$, $G_2$ is replicated $n_0$ times by the bounded spine $n_0X$ of $\overline{n_0}$. This iterates $n_0$ times the function $g'$, leading to a new function $g'^{(n_0)}$.
Now, the behavior of the net becomes obvious.
$g'^{(n_0)}$ takes $\tuple{0,f}$ and $\overrightarrow n$ as its arguments.
Then, it iterates the step function. Finally, the result has form $\tuple{n_0,g'^{(n)}\left(\tuple{n_0,f},\overrightarrow n\right)}$.
An application of the right projection gives the result.

\begin{proof}[Proof of Theorem \ref{theorem:Completeness}.]
By induction on the structure of $\funone$. The possible structures to be considered are: 0, $\Succ$, $\pi^m_i$, $\prcom{\funthree}{\funtwo_1,\ldots,\funtwo_n}$ and $\prrec{\funthree}{\funtwo}$.
If $\funone\equiv 0$ (base case), the statement holds by the reflexivity of $\recutel^*$.
If $\funone\equiv\Succ$ or $\pi^m_i$, we already know that these functions can be represented.
If $\funone\equiv\prcom{\funthree}{\funtwo_1,\ldots,\funtwo_n}$, by inductive hypothesis the functions $\prtopn{\funthree}$ and $\prtopn{\funtwo_i}$ can be represented. And the composition is representable.
Similarly, if $\funone\equiv\prrec{\funthree}{\funtwo}$, by inductive hypothesis the functions $\prtopn{\funthree}$ and $\prtopn{\funtwo}$ can be represented. So their primitive recursion $\prrec{\funthree}{\funtwo}$ is representable too.
\end{proof}
}
{
Our comments at the end of Section~\ref{section:Pure proof nets} should have convinced the reader about
the impossibility of proving completeness by the usual encoding of the untyped lambda-calculus into pure proof nets.
We really need to tailor the encoding of data and programs in such a way that superlazy reduction \emph{works} on
them, i.e., we want to be sure that a program applied to an argument \emph{reduces} to the intended result.

The details of the completeness proof do not fit here, but they can be found in~\cite{longversion}, where we also also develop a slightly more general proof of soundness. Only the
major ingredients towards completeness are reported here.

The first ingredient is the representation of natural numbers.
Figure~\ref{fig:ChurchNumeralAsPN} introduces the closed proof nets $\overline{0}$ and
$\overline{n}$, with $n\geq 1$.
\begin{figure}[ht]
\begin{center}
  \begin{minipage}[c]{27pt}
    \centering\scalebox{0.5}{\epsfbox{figure.1041}}
  \end{minipage}
  \begin{minipage}[c]{10pt}
    \centering $\equiv$
  \end{minipage}
  \begin{minipage}[c]{33pt}
    \centering\scalebox{0.5}{\epsfbox{figure.1051}}
  \end{minipage}
  \hspace{10pt}
\hspace{30pt}
  \hspace{10pt}
  \begin{minipage}[c]{59pt}
    \centering\scalebox{0.5}{\epsfbox{figure.104}}
  \end{minipage}
  \begin{minipage}[c]{10pt}
    \centering $\equiv$
  \end{minipage}
  \begin{minipage}[c]{33pt}
    \centering\scalebox{0.5}{\epsfbox{figure.105}}
  \end{minipage}
\caption{Church numerals as proof nets}
\label{fig:ChurchNumeralAsPN}
\end{center}
\end{figure}
A basic observation to prove the completeness is that every $\overline{n}$ contains a
$n$-bounded spine $nX$ (see Fig.~\ref{fig:Derelicting tree-b} ).
This implies that whenever $\overline{n}$ is applied to a box containing a proof net
$G$, superlazy reduction will copy the box $n$ times and $n$ copies of $G$ will appear,
one applied to the next one. Computing the successor of a natural number can be done by
a proof net, quite similarly to what happens in ordinary lambda calculus.
Now, at least, the key notion can be given. Let $G$ be a proof net with a single premise and a single conclusion.
We write $G\smile\tuple{\overline{n_1},\ldots,\overline{n_m}}$ to mean the closed proof net with
a single conclusion obtained by plugging the conclusion of the net
$\tuple{\overline{n_1},\ldots,\overline{n_m}}$ (obtained by ``tensoring'' together $\overline{n_1},\ldots,\overline{n_m}$)
into the unique assumption of $G$.
For every function $\funone:\N^m\rightarrow\N$, with arity $m\geq 0$, we shall say
that a proof net $\prtopn{\funone}$ with one premise \emph{represents} $\funone$ iff
$\prtopn{\funone}\smile\tuple{\overline{n_1},\ldots,\overline{n_m}}$ \emph{superlazily} reduces
to $\overline{\funone(n_1,\ldots,n_m)}$, for every $n_1,\ldots,n_m\in\N$.\par
The second ingredient is the possibility of freely duplicate data, i.e., natural numbers.
This is possible even if the natural number being copied (or erased) does not lie inside
a box.
Copying a natural number $n$ involves applying a (boxed) pair of successors and a pair
of $\overline{0}$ to the net $\overline{n}$.
Erasing $n$, on the other hand, can be
performed by applying a boxed identity and another boxed identity to $\overline{n}$: this
way $\overline{n}$ can be reduced itself to a boxed identity, which can be erased by
cutting it against a $W$ node.

The first two ingredients allow to encode basic primitive recursive functions and
composition. To get the most important construction, namely primitive recursion itself,
a third ingredient is necessary, namely iteration. Iterating $n$ times a given function $f$,
where $\overline{n}$ is a parameter, can be done by way of superlazy reduction:
putting the proof net representing $f$ inside a box and apply the
box to $\overline{n}$ suffices. However, if the proof net representing $f$ has some
premises, the resulting proof net would only accept boxed natural numbers as arguments,
and this would break the scheme which makes superlazy reduction works. The solution consists
in iterating only closed functions, exploiting the higher-order nature of proof nets.
The usual primitive recursion scheme can be finally obtained by using the standard technique of
simulating recursion by iteration.

\begin{theorem}[Completeness]
\label{theorem:Completeness}
Every $\funone$ in \PR\ is represented by a proof net $\prtopn{\funone}$.
\end{theorem}
}

%% file: extensions.tex
\section{Further Developments}\label{section:FurtherDevelopments}
As we explained in the Introduction, this paper is just the first
step in a long-term study about how to control modal impredicativity.

The are at least two distinct research directions the authors are
following at the time of writing. We recall the obtained results here,
pointing to further work for additional details and proofs.

First of all, the way we proved completeness of proof net reduction w.r.t. primitive
recursion suggests a way of capturing Hofmann's non-size increasing
computation~\cite{hofmann99lics} by way of a proper, further restriction to
superlazy reduction. We basically need two constraints:
\begin{varitemize}
\item
  Boxes can only be copied when they are closed, i.e., when they have no premises.
\item
  Any box $b$ can only be copied if the proof net contained in $b$
  only contains \emph{one} node $X$ at level $0$.
\end{varitemize}
The obtained reduction relation is said to be the \emph{non-size increasing superlazy
reduction} and is denoted with $\Rightarrow$.
With these constraints, non-size increasing
polytime computation can be simulated by pure proof nets. Moreover:
\begin{theorem}
For every $n\in\N$ there is a polynomial
$p_n:\N\rightarrow\N$ such that for every proof net $\Pi$,
$\dlength{G}{\Rightarrow},\dsize{G}{\Rightarrow}\leq p_{\partial(G)}(|G|)$
\end{theorem}

The second direction we are considering concerns methods to keep modal
impredicativity under control by more traditional static methods in the spirit
of light logics. Consider the equivalence between $!A$ and $!A\otimes A$.
As already pointed out, this equivalence is somehow necessary to get modal
impredicativity. So, controlling it means controlling modal impredicativity.
Now, suppose that the above equivalence holds \emph{but} $A$ only contains
instances of the $!$ operator which are \emph{intrinsically different} from the top-level one
in $!A$. Morally, this would imply that even if a contraction node is ``in'' $A$,
it cannot communicate with the box in $!A$, because they are of a different nature.
This way modal impredicativity would be under control. But the question is: how
to distinguish different instances of $!$ from each other? One (naive) answer
is the following: consider a generalization of (multiplicative and exponential)
Linear Logic where syntactically different copies $!_{a_1},!_{a_2},\ldots$
of the modal operator $!$ are present.
As an example, take the set $\{!_n\}_{n\in\N}$. Then, impose the following constraint:
$!_aA$ is a \emph{legal} formula only if the operators $!_{b_1},\ldots,!_{b_n}$ appearing
in $A$ are all different from $a$, i.e., if $a\not\in\{b_1,\ldots,b_n\}$. We strongly
believe that this way a system enjoying properties similar to those of predicative
recurrence schemes~\cite{Leivant93popl} can be obtained.

%% file: conclusions.tex
\section{Conclusions}
\label{section:Conclusions}
We described modal impredicativity and a concrete
tool --- superlazy reduction --- that controls it.
Superlazy reduction on pure proof nets
greatly influences the expressive power of the underlying computational
model: from a Turing complete model we go down to first-order primitive
recursive functions.

In a sentence, we learn that the expressive
power of a programming language system can be controlled by acting on the
dynamics (i.e., the underlying reduction relation) without touching the
statics (i.e., the language into which programs are written). To get the complete picture, however, 
we still need tools to predict which set of programs will be useful from a computational
point of view.

This could have potential applications in the field of implicit computational
complexity, where one aims at designing programming languages and logical
systems corresponding to complexity classes. Indeed, the impact of ICC in
applications crucially depends on the intensional expressivity of the
proposed systems: one should be able to write programs naturally.
